\documentclass[11pt]{article}

\usepackage{graphicx}
\usepackage{latexsym}
\usepackage{enumerate}
\usepackage{amssymb}
\usepackage{ifpdf}
\usepackage{amsmath}
\usepackage{bbm}
\PassOptionsToPackage{boxed,section}{algorithm}
\usepackage{algorithm}
\usepackage{algorithmic}
\usepackage{times}
\usepackage{tikz}
\usetikzlibrary{arrows}
\usetikzlibrary{shapes}
\usetikzlibrary{decorations.markings}
\usepackage{verbatim}
\usepackage{setspace}

\usepackage{tikz}
\usetikzlibrary{arrows}
\usetikzlibrary{shapes}

\newcommand{\FUNCTION}[0]{{\bf Function}}
\newtheorem{theorem}{\bf Theorem}[section]
\newenvironment{proof}{ {\bf Proof.}} {$\Box$}
\newcommand*\samethanks[1][\value{footnote}]{\footnotemark[#1]}

\title{On the Shoshan-Zwick Algorithm for the All-Pairs Shortest
Path Problem}

\author{
Pavlos Eirinakis \thanks{This research has been co-financed by the European Union (European Social Fund - ESF) and Greek 
national funds through the Operational Program "Education and Lifelong Learning" of the National Strategic 
Reference Framework (NSRF) - Research Funding Program: Thales. Investing in knowledge society through the 
European Social Fund. (MIS: 380232)}\\
Athens University of Economics and Business,\\
Athens, Greece\\
\mbox{\{peir@aueb.gr\}}
\and
Matthew Williamson \thanks{This research was supported in part by the National Science Foundation and Air Force of Scientific 
Research through Awards CNS-0849735, CCF-1305054, and FA9550-12-1-0199.}\\
West Virginia University Institute of Technology\\
Montgomery, WV, USA\\
\mbox{\{matthew.williamson@mail.wvu.edu\}}
\and
 K. Subramani \samethanks\\
West Virginia University, \\
Morgantown, WV, USA \\
\mbox{\{ksmani@csee.wvu.edu\}}
}

\begin{document}
\date{}
\maketitle
\thispagestyle{myheadings}
\bibliographystyle{plain}

\onehalfspacing

\begin{abstract}
The Shoshan-Zwick algorithm solves the 
 all pairs shortest paths problem in undirected graphs with integer edge costs in 
 the range $\{1, 2, \dots, M\}$.  It runs in $\tilde{O}(M\cdot n^{\omega})$ time, 
 where $n$ is the number of vertices, $M$ is the largest integer edge cost, 
 and $\omega < 2.3727$ is the exponent of matrix multiplication.  
It is the fastest known algorithm  for this problem.
This paper points out the erroneous behavior of the Shoshan-Zwick algorithm and revises the algorithm to resolve 
the issues that cause this behavior.
Moreover, it discusses implementation aspects of the Shoshan-Zwick algorithm using currently-existing sub-cubic matrix multiplication algorithms.
\end{abstract}

\section{Introduction}
\label{sec:intro}

   The Shoshan-Zwick algorithm \cite{Shosh1} 
  solves the all-pairs shortest paths (APSP) problem in undirected graphs, where the edge 
  costs are integers in the range $\{1, 2, \dots, M\}$.  This is accomplished by computing $O(\log(M\cdot n))$ 
  distance products of $n \times n$ matrices with elements in the range $\{1, 2, \dots, M\}$.  The 
  algorithm runs in $\tilde{O}(M\cdot n^{\omega})$ time, where $\omega < 2.3727$ is the exponent 
  for the fastest known matrix multiplication algorithm \cite{Williams12}.  
  
    The APSP problem is a fundamental problem in algorithmic graph theory.  
  Consider a graph with $n$ nodes and $m$ edges.  
  For directed or undirected graphs with real edge costs, 
  we can use known methods that run in $O(m\cdot n + n^{2}\cdot \log n)$ time 
  \cite{dijkstra59,fredman87,johnson77} and $O(n^3)$ time \cite{floyd1}.  Sub-cubic APSP 
  algorithms have been obtained that run in $O(n^{3}\cdot ((\log \log n)/\log n)^{1/2})$ 
  time \cite{Fredman76}, $O(n^{3}\cdot \sqrt{\log \log n/\log n})$ time \cite{Takaoka92}, 
  $O(n^{3}/\sqrt{\log n})$ time \cite{Dobo90}, $O(n^{3}\cdot (\log \log n/\log n)^{5/7})$ 
  time \cite{Han04}, $O(n^{3}\cdot (\log \log n)^{2}$ $\log n)$ time \cite{Takaoka04}, 
  $O(n^{3}\cdot \log \log n/$ $\log n)$ time \cite{Takaoka05}, 
  $O(n^{3}\cdot \sqrt{\log \log n}/$ $\log n)$ time \cite{Zwick06}, $O(n^{3}/\log n)$ time 
  \cite{Chan08}, and $O(n^{3}\cdot \log \log n/\log^{2} n)$ time \cite{HanTakaoka12}.  
  For undirected graphs with integer edge costs, the problem can be 
  solved in $O(m\cdot n)$ time \cite{Thorup97,Thorup98}.  Fast matrix multiplication algorithms 
  can also be used for solving the APSP problem in dense graphs with small integer edge costs.  
  \cite{GM97,Raim1} provide algorithms that run in $\tilde{O}(n^{\omega})$ time in unweighted, 
  undirected graphs, where $\omega < 2.3727$ is the exponent for matrix multiplication \cite{Williams12}.  
  \cite{Chan12} improves this result with an algorithm that runs in $o(m\cdot n)$ time.
  
  In this paper, we revise the Shoshan-Zwick algorithm to resolve some issues that cause the algorithm to 
  behave erroneously. This behavior was identified when, in the process of implementing the 
  algorithm as part of a more elaborate procedure 
  (i.e., identifying negative cost cycles in undirected graphs),
  we discovered that the algorithm is not producing correct results.  Since the Shoshan-Zwick algorithm is incorrect in its current version, any result that uses this algorithm as a subroutine is also incorrect.  For instance, the results provided by Alon and Yuster \cite{AY07}, Cygan et. al \cite{CGS12}, W. Liu el al \cite{LWJLW12},  and Yuster \cite{Yuster11} all depend on the Shoshan-Zwick algorithm.  Thus, their results are currently incorrect.  By applying our revision to the algorithm, the issues with the above results are resolved.  We also discuss issues concerning the implementation of the Shoshan-Zwick algorithm using known sub-cubic matrix multiplication algorithms.

  The principal contributions of this paper are as follows:
  \begin{enumerate}[(i)]
  \item A counter-example that shows that the Shoshan-Zwick algorithm is incorrect  in its current form.
  \item A detailed explanation of where and why the algorithm fails.
  \item A modified version of the Shoshan-Zwick algorithm that corrects the problems with the previous version.  The 
corrections do not affect the time complexity of the algorithm.
  \item A discussion concerning implementing the matrix multiplication subroutine that is used in the algorithm.

 \end{enumerate}
 
   The rest of the paper is organized as follows.  We describe the Shoshan-Zwick algorithm in 
   Section~\ref{alg}.  Section \ref{counter} establishes
   that the published version of the algorithm is incorrect by providing a
   simple counter-example.  
   The origins of this erroneous behavior is identified in Section \ref{error}. 
   In Section \ref{corr}, we present a revised
   version of the algorithm and provide a formal proof of
   correctness.   In Section \ref{eff}, we discuss the efficacy of the Shoshan-Zwick
   algorithm.  
Section \ref{conc} concludes the paper by summarizing our contributions and discussing avenues for future
  research. 
   
   \section{The Shoshan-Zwick Algorithm}
   \label{alg}
   
   In this section, we review the Shoshan-Zwick algorithm for the APSP problem in undirected 
   graphs with integer edge costs.  Consider a graph $G = (V, E)$, where $V = \{0, 1, 2, \ldots, n\}$ is the set of 
   nodes, and $E$ is the set of edges.  The graph is represented as an $n \times n$ matrix 
   ${\bf D}$, where $d_{ij}$ is the cost of edge $(i, j)$ if $(i,j) \in E$, $d_{ii} = 0$ for $1 \leq i\leq n$, and 
   $d_{ij} = +\infty$ otherwise. 
   Note that, without loss of generality, the edge costs are taken from the range $\{1, 2, \dots, M\}$, where 
   $M = 2^{m}$ for some $m \geq 1$.
   
   The algorithm computes a logarithmic number of distance products in order to determine 
   the shortest paths.  Let ${\bf A}$ and ${\bf B}$ be two $n \times n$ 
   matrices. Their distance product ${\bf A} \star {\bf B}$ is an $n \times n$ matrix such that:
   
   \[({\bf A} \star {\bf B})_{ij} = \min_{k=1}^{n}\{a_{ik} + b_{kj}\}, 1 \leq i,j \leq n.\]
   
   \noindent The distance product of two $n \times n$ matrices with elements within 
   the range $\{1, 2, \dots,\allowbreak M, +\infty\}$ can be computed in $\tilde{O}(M\cdot n^{\omega})$ 
   time \cite{AGM97}.
   The distance product of two matrices whose finite elements are taken from $\{1, 2, \dots, M\}$ 
   is a matrix whose finite elements are taken from $\{1, 2 \dots, 2\cdot M\}$.  
   However, the Shoshan-Zwick algorithm is based on not allowing the range 
   of elements in the matrices it uses to increase. Hence, if ${\bf A}$ 
   is an $n \times n$ matrix, and $a, b$ are two numbers such that $a\leq b$,
   the algorithm utilizes two operations, namely $clip({\bf A}, a, b)$ and 
   $chop({\bf A}, a, b)$, that produce  corresponding $n \times n$ matrices such that  
    \[(clip({\bf A}, a, b))_{ij} = \left\{\begin{array}{ll}
 a & \mbox{ if } a_{ij} < a\\     
 a_{ij} & \mbox{ if } a \leq a_{ij} \leq b\\
 +\infty & \mbox{ if } a_{ij} > b\\
\end{array}\right.\]

 \[(chop({\bf A}, a, b))_{ij} = \left\{\begin{array}{ll}
 a_{ij} & \mbox{ if } a \leq a_{ij} \leq b\\
 +\infty & \mbox{ otherwise}\\
\end{array}\right..\]

  The algorithm also defines $n \times n$ matrices ${\bf A} \bigwedge {\bf B}$, ${\bf A} 
  \bar{\bigwedge} {\bf B}$, and ${\bf A} \bigvee {\bf B}$ such that
\begin{align*}
 \left({\bf A} \bigwedge {\bf B}\right)_{ij} = &  \left\{\begin{array}{ll}
 a_{ij} & \mbox{ if } b_{ij} < 0\\
 +\infty & \mbox{ otherwise}\\
\end{array}\right.
\\
\left({\bf A} \bar{\bigwedge} {\bf B}\right)_{ij} = &  \left\{\begin{array}{ll}
 a_{ij} & \mbox{ if } b_{ij} \geq 0\\
 +\infty & \mbox{ otherwise}\\
\end{array}\right.
\\
\left({\bf A} \bigvee {\bf B}\right)_{ij} = & \left\{\begin{array}{ll}
 a_{ij} & \mbox{ if } a_{ij} \neq +\infty\\
 b_{ij} & \mbox{ if } a_{ij} = +\infty, b_{ij} \neq +\infty\\
 +\infty & \mbox{ if } a_{ij} = b_{ij} = +\infty\\
\end{array}\right..
\end{align*}
 
Finally, if ${\bf C} = (c_{ij})$ and ${\bf P} = (p_{ij})$ are matrices, the algorithm 
defines the Boolean matrices $({\bf C} \geq 0)$ and $(0 \leq {\bf P} \leq M)$ 
such that 

\[\left({\bf C} \geq 0\right)_{ij} = \left\{\begin{array}{ll}
 1 & \mbox{ if } c_{ij} \geq 0\\
 0 & \mbox{ otherwise}\\
\end{array}\right.\]

\[\left(0 \leq {\bf P} \leq M\right)_{ij} = \left\{\begin{array}{ll}
 1 & \mbox{ if } 0 \leq p_{ij} \leq M\\
 0 & \mbox{ otherwise}\\
\end{array}\right.\]

  Using the definitions above, the Shoshan-Zwick algorithm computes the shortest 
  paths by  finding the $\lceil \log_2 n\rceil$ most significant 
  bits of each distance and the remainder that each distance leaves 
  modulo $M$.  This provides enough information to reconstruct the distances.  The 
  procedure is shown in Algorithm~\ref{alg:sz}, and the proof of correctness is given 
  in \cite[Lemma 3.6]{Shosh1}.
 
The algorithm computes $O(\log n + \log M) = 
  O(\log (M\cdot n))$ distance products of matrices, whose finite elements are in the range 
  $\{0, 1, 2, \dots, 2\cdot M\}$ or in the range $\{-M, \dots, M\}$.  Note that each 
  distance product can be reduced to a constant number of distance products of matrices
  with elements in the range $\{1, 2, \dots, M\}$.  All other operations in the 
  algorithm take $O(n^{2}\cdot (\log n + \log M))$ time.  Therefore, the total running 
  time of the algorithm is $\tilde{O}(M\cdot n^{\omega})$ time \cite[Theorem 3.7]{Shosh1}.
 
 \begin{algorithm}[ht!]
  \FUNCTION{ {\sc Shoshan-Zwick-APSP}$({\bf D})$}
  \begin{algorithmic}[1]
    \STATE{$l = \lceil \log_{2} n\rceil$.}
    \STATE{$m = \log_{2} M$.}
    \FOR{$(k = 1$ to $m+1)$}
      \STATE{${\bf D} = clip({\bf D} \star {\bf D}, 0, 2\cdot M)$.}
    \ENDFOR
    \STATE{${\bf A}_0 = {\bf D} - M$.}
    \FOR{$(k = 1$ to $l)$}
      \STATE{${\bf A}_k = clip({\bf A}_{k-1} \star {\bf A}_{k-1}, -M, M)$.}
    \ENDFOR
    \STATE{${\bf C}_l = -M$.}
    \STATE{${\bf P}_l = clip({\bf D}, 0, M)$.}
    \STATE{${\bf Q}_l = +\infty$.}
    \FOR{$(k = l-1$ down to $0)$}
      \STATE{${\bf C}_k = [clip({\bf P}_{k+1} \star {\bf A}_k, -M, M) \bigwedge {\bf C}_{k+1}] 
	\bigvee [clip({\bf Q}_{k+1} \star {\bf A}_k, -M, M)$ $\bar{\bigwedge} {\bf C}_{k+1}]$.}
      \STATE{${\bf P}_k = {\bf P}_{k+1} \bigvee {\bf Q}_{k+1}$.}
      \STATE{${\bf Q}_k =chop({\bf C}_k, 1-M, M)$.}
    \ENDFOR
    \FOR{$(k = 1$ to $l)$}
      \STATE{${\bf B}_{k} = ({\bf C}_k \geq 0)$.}
    \ENDFOR
    \STATE{${\bf B}_0 = (0 \leq {\bf P}_0 < M)$.}
    \STATE{${\bf R} = {\bf P}_0 \mod M$.}
    \STATE{$\mathrm{\Delta} = M\cdot \sum_{k=0}^{l} 2^{k}\cdot {\bf B}_k + {\bf R}$.}
    \RETURN $\mathrm{\Delta}$.
  \end{algorithmic}
  \caption{Shoshan-Zwick APSP Algorithm}
  \label{alg:sz}
\end{algorithm}

  \section{A Counter-Example}
  \label{counter}
  
  In this section, we provide a detailed presentation of the counter-example presented that shows that the Shoshan-Zwick algorithm is incorrect as presented in \cite{Shosh1}. 
Consider the APSP problem with respect to graph $G'$ presented in Figure~\ref{fig:example1}. 

    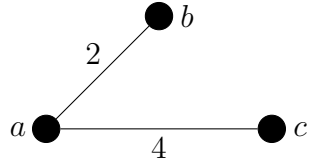
\begin{figure}[htbp]
 \begin{center}
  \begin{tikzpicture}[scale=0.5]
    \node[fill=black,circle,draw] at (-3,0) (a) {};
    \node[fill=black,circle,draw] at (0,3) (b) {};
    \node[fill=black,circle,draw] at (3,0) (c) {};
    
    \node at (-3.75,0) {{\large $a$}};
    \node at (0.75,3) {{\large $b$}};
    \node at (3.75,0) {{\large $c$}};
    
    \draw (a)--(b);
    \draw (a)--(c);
    
    \node at (0,-0.5) {{\large $4$}};
    \node at (-1.75,2) {{\large $2$}};
  \end{tikzpicture}
 \end{center}
\caption{Graph $G'$ of the counter-example for the Shoshan-Zwick algorithm.}
\label{fig:example1}       
\end{figure}

Graph $G'$ can also be represented as a $3 \times 3$ matrix:
  \[{\bf D} = 
   \left[\begin{array}{ccc}
   0 & 2 & 4\\
   2 & 0 & \infty\\
   4 & \infty & 0\\
   \end{array}\right].
  \]
  We now walk through each step of the algorithm.  First, $l = \lceil \log_2 3\rceil = 2$, and 
  $m = \log_2 4 = 2$.  This means that in the {\bf for} loop in lines $3$ to $5$, we set 
  ${\bf D} = clip({\bf D}\star{\bf D}, 0, 8)$ three times.  At each such step, we 
  get 
  {\scriptsize
\begin{align*}
{\bf D} \star {\bf D} & = 
  \left[\begin{array}{ccc}
   0 & 2 & 4\\
   2 & 0 & \infty\\
   4 & \infty & 0\\
  \end{array}\right] \star 
  \left[\begin{array}{ccc}
   0 & 2 & 4\\
   2 & 0 & \infty\\
   4 & \infty & 0\\
  \end{array}\right] = 
  \left[\begin{array}{rrr}
   0 & 2 & 4\\
   2 & 0 & 6\\
   4 & 6 & 0\\
  \end{array}\right],
  & {\bf D} = 
  \left[\begin{array}{rrr}
  0 & 2 & 4\\
   2 & 0 & 6\\
   4 & 6 & 0\\
  \end{array}\right] (\text{for }k=1),
  \\
    {\bf D} \star {\bf D} & = 
  \left[\begin{array}{rrr}
  0 & 2 & 4\\
  2 & 0 & 6\\
  4 & 6 & 0\\
  \end{array}\right] \star 
  \left[\begin{array}{rrr}
  0 & 2 & 4\\
  2 & 0 & 6\\
  4 & 6 & 0\\
  \end{array}\right] = 
  \left[\begin{array}{rrr}
  0 & 2 & 4\\
  2 & 0 & 6\\
  4 & 6 & 0\\
  \end{array}\right],
   &  {\bf D} = 
  \left[\begin{array}{rrr}
  0 & 2 & 4\\
   2 & 0 & 6\\
   4 & 6 & 0\\
  \end{array}\right] (\text{for }k=2),
  \\
    {\bf D} \star {\bf D} & = 
  \left[\begin{array}{rrr}
  0 & 2 & 4\\
  2 & 0 & 6\\
  4 & 6 & 0\\
  \end{array}\right] \star 
  \left[\begin{array}{rrr}
  0 & 2 & 4\\
  2 & 0 & 6\\
  4 & 6 & 0\\
  \end{array}\right] = 
  \left[\begin{array}{rrr}
  0 & 2 & 4\\
  2 & 0 & 6\\
  4 & 6 & 0\\
  \end{array}\right],
   &  {\bf D} = 
  \left[\begin{array}{rrr}
  0 & 2 & 4\\
   2 & 0 & 6\\
   4 & 6 & 0\\
  \end{array}\right] (\text{for } k=3).  
\end{align*}  
  }
  The next step is to set ${\bf A}_0 = {\bf D} - {\bf M}$, which gives us 
  \[{\bf A}_0 = 
  \left[\begin{array}{rrr}
  0 & 2 & 4\\
  2 & 0 & 6\\
  4 & 6 & 0\\
  \end{array}\right] - 4 = 
  \left[\begin{array}{rrr}
  -4 & -2 & 0\\
  -2 & -4 & 2\\
  0 & 2 & -4\\
  \end{array}\right].\]
  
  In the {\bf for} loop in lines $7$ to $9$, we compute ${\bf A}_k = clip({\bf A}_{k-1} \star {\bf A}_{k-1}, -4, 4)$ 
  for $k=1$ and $k=2$. This gives us 
  {\scriptsize
  \begin{align*}
  {\bf A}_0 \star {\bf A}_0 = 
  \left[\begin{array}{rrr}
  -4 & -2 & 0\\
  -2 & -4 & 2\\
  0 & 2 & -4\\
  \end{array}\right] \star 
  \left[\begin{array}{rrr}
  -4 & -2 & 0\\
  -2 & -4 & 2\\
  0 & 2 & -4\\
  \end{array}\right] = 
  \left[\begin{array}{rrr}
  -8 & -6 & -4\\
  -6 & -8 & -2\\
  -4 & -2 & -8\\
  \end{array}\right],
  & \; {\bf A}_1 = 
  \left[\begin{array}{rrr}
  -4 & -4 & -4\\
  -4 & -4 & -2\\
  -4 & -2 & -4\\
  \end{array}\right],
\\
    {\bf A}_1 \star {\bf A}_1 = 
  \left[\begin{array}{rrr}
  -4 & -4 & -4\\
  -4 & -4 & -2\\
  -4 & -2 & -4\\
  \end{array}\right] \star 
  \left[\begin{array}{rrr}
  -4 & -4 & -4\\
  -4 & -4 & -2\\
  -4 & -2 & -4\\
  \end{array}\right] = 
  \left[\begin{array}{rrr}
  -8 & -8 & -8\\
  -8 & -8 & -8\\
  -8 & -8 & -8\\
  \end{array}\right],
  & \; {\bf A}_2 = 
  \left[\begin{array}{rrr}
  -4 & -4 & -4\\
  -4 & -4 & -4\\
  -4 & -4 & -4\\
  \end{array}\right].  
  \end{align*}  
  }
  
  In lines $10$ to $12$, we set ${\bf C}_2 = -4$, ${\bf P}_2 = clip({\bf D}, 0, 4)$, 
  and ${\bf Q}_2 = +\infty$. Thus, we have
  \[{\bf C}_2 = 
  \left[\begin{array}{rrr}
  -4 & -4 & -4\\
  -4 & -4 & -4\\
  -4 & -4 & -4\\
  \end{array}\right], 
  {\bf P}_2 = 
  \left[\begin{array}{ccc}
  0 & 2 & 4\\
  2 & 0 & \infty\\
  4 & \infty & 0\\
  \end{array}\right],  
  {\bf Q}_2 = 
  \left[\begin{array}{ccc}
  \infty & \infty & \infty\\
  \infty & \infty & \infty\\
  \infty & \infty & \infty\\
  \end{array}\right].\]
  
  We now compute the {\bf for} loop in lines $13$ to $17$.  This means we run the lines 
  \[{\bf C}_k = [clip({\bf P}_{k+1} \star {\bf A}_k, -4, 4) \bigwedge {\bf C}_{k+1}] 
	\bigvee [clip({\bf Q}_{k+1} \star {\bf A}_k, -4, 4) \bar{\bigwedge} {\bf C}_{k+1}],\] 
  \[{\bf P}_k = {\bf P}_{k+1} \bigvee {\bf Q}_{k+1}, \mbox{ and}\] \[{\bf Q}_k =chop({\bf C}_k, -3, 4)\] twice,
  i.e., for $k=1$ and $k=0$.  
  After this loop, we get 
 {\scriptsize
  \[(\text{for }k=1) \;\;\; {\bf P}_2 \star {\bf A}_1 = 
  \left[\begin{array}{ccc}
  0 & 2 & 4\\
  2 & 0 & \infty\\
  4 & \infty & 0\\
  \end{array}\right] \star 
  \left[\begin{array}{rrr}
  -4 & -4 & -4\\
  -4 & -4 & -2\\
  -4 & -2 & -4\\
  \end{array}\right] = 
  \left[\begin{array}{rrr}
  -4 & -4 & -4\\
  -4 & -4 & -2\\
  -4 & -2 & -4\\
  \end{array}\right],
  \]
  \[
  clip({\bf P}_2 \star {\bf A}_1, -4, 4) = 
  \left[\begin{array}{rrr}
  -4 & -4 & -4\\
  -4 & -4 & -2\\
  -4 & -2 & -4\\
  \end{array}\right],
  \]
    \[
  clip({\bf P}_2 \star {\bf A}_1, -4, 4) \bigwedge {\bf C}_{2} = 
  \left[\begin{array}{rrr}
  -4 & -4 & -4\\
  -4 & -4 & -2\\
  -4 & -2 & -4\\
  \end{array}\right] 
  \bigwedge
    \left[\begin{array}{rrr}
  -4 & -4 & -4\\
  -4 & -4 & -4\\
  -4 & -4 & -4\\
  \end{array}\right] =
  \left[\begin{array}{rrr}
  -4 & -4 & -4\\
  -4 & -4 & -2\\
  -4 & -2 & -4\\
  \end{array}\right],
  \]
  \[{\bf Q}_2 \star {\bf A}_1 = 
  \left[\begin{array}{ccc}
  \infty & \infty & \infty\\
  \infty & \infty & \infty\\
  \infty & \infty & \infty\\
  \end{array}\right] \star 
  \left[\begin{array}{rrr}
  -4 & -4 & -4\\
  -4 & -4 & -2\\
  -4 & -2 & -4\\
  \end{array}\right] = 
  \left[\begin{array}{rrr}
  \infty & \infty & \infty\\
  \infty & \infty & \infty\\
  \infty & \infty & \infty\\
  \end{array}\right],
  \]
   \[
  clip({\bf Q}_2 \star {\bf A}_1, -4, 4) = 
  \left[\begin{array}{rrr}
  \infty & \infty & \infty\\
  \infty & \infty & \infty\\
  \infty & \infty & \infty\\
  \end{array}\right],
  \]
    \[
  clip({\bf Q}_2 \star {\bf A}_1, -4, 4) \bar{\bigwedge} {\bf C}_{2} = 
  \left[\begin{array}{rrr}
  \infty & \infty & \infty\\
  \infty & \infty & \infty\\
  \infty & \infty & \infty\\
  \end{array}\right] 
  \bar{\bigwedge}
    \left[\begin{array}{rrr}
  -4 & -4 & -4\\
  -4 & -4 & -4\\
  -4 & -4 & -4\\
  \end{array}\right] =
  \left[\begin{array}{rrr}
  \infty & \infty & \infty\\
  \infty & \infty & \infty\\
  \infty & \infty & \infty\\
  \end{array}\right],
  \]
  \[{\bf C}_1 = 
    \left[\begin{array}{rrr}
  -4 & -4 & -4\\
  -4 & -4 & -2\\
  -4 & -2 & -4\\
  \end{array}\right]
  \bigvee
    \left[\begin{array}{rrr}
  \infty & \infty & \infty\\
  \infty & \infty & \infty\\
  \infty & \infty & \infty\\
  \end{array}\right] =
  \left[\begin{array}{rrr}
  -4 & -4 & -4\\
  -4 & -4 & -2\\
  -4 & -2 & -4\\
  \end{array}\right], 
  \]
    \[{\bf P}_1 = {\bf P}_2 \bigvee {\bf Q}_2 = 
  \left[\begin{array}{ccc}
  0 & 2 & 4\\
  2 & 0 & \infty\\
  4 & \infty & 0\\
  \end{array}\right]
  \bigvee
    \left[\begin{array}{rrr}
  \infty & \infty & \infty\\
  \infty & \infty & \infty\\
  \infty & \infty & \infty\\
  \end{array}\right] =
  \left[\begin{array}{ccc}
  0 & 2 & 4\\
  2 & 0 & \infty\\
  4 & \infty & 0\\
  \end{array}\right], 
  \]
  \[
  {\bf Q}_1 = chop({\bf C}_1, -3, 4) = 
  \left[\begin{array}{rrr}
  \infty & \infty & \infty\\
  \infty & \infty & -2\\
  \infty & -2 & \infty\\
  \end{array}\right],\]
    \[(\text{and for }k=0) \;\;\; {\bf P}_1 \star {\bf A}_0 = 
  \left[\begin{array}{ccc}
  0 & 2 & 4\\
  2 & 0 & \infty\\
  4 & \infty & 0\\
  \end{array}\right] \star 
  \left[\begin{array}{rrr}
  -4 & -2 & 0\\
  -2 & -4 & 2\\
  0 & 2 & -4\\
  \end{array}\right] = 
  \left[\begin{array}{rrr}
  -4 & -2 & 0\\
  -2 & -4 & 2\\
  0 & 2 & -4\\
  \end{array}\right],
  \]
  \[
  clip({\bf P}_1 \star {\bf A}_0, -4, 4) = 
  \left[\begin{array}{rrr}
  -4 & -2 & 0\\
  -2 & -4 & 2\\
  0 & 2 & -4\\
  \end{array}\right],
  \]
    \[
  clip({\bf P}_1 \star {\bf A}_0, -4, 4) \bigwedge {\bf C}_{1} = 
  \left[\begin{array}{rrr}
  -4 & -2 & 0\\
  -2 & -4 & 2\\
  0 & 2 & -4\\
  \end{array}\right] 
  \bigwedge
    \left[\begin{array}{rrr}
  -4 & -4 & -4\\
  -4 & -4 & -2\\
  -4 & -2 & -4\\
  \end{array}\right] =
  \left[\begin{array}{rrr}
  -4 & -2 & 0\\
  -2 & -4 & 2\\
  0 & 2 & -4\\
  \end{array}\right],
  \]
  \[{\bf Q}_1 \star {\bf A}_0 = 
  \left[\begin{array}{ccc}
  \infty & \infty & \infty\\
  \infty & \infty & -2\\
  \infty & -2 & \infty\\
  \end{array}\right] \star 
  \left[\begin{array}{rrr}
  -4 & -2 & 0\\
  -2 & -4 & 2\\
  0 & 2 & -4\\
  \end{array}\right] = 
  \left[\begin{array}{rrr}
  \infty & \infty & \infty\\
  -2 & 0 & -6\\
  -4 & -6 & -4\\
  \end{array}\right],
  \]
   \[
  clip({\bf Q}_1 \star {\bf A}_0, -4, 4) = 
  \left[\begin{array}{rrr}
  \infty & \infty & \infty\\
  -2 & 0 & -4\\
  -4 & -4 & -4\\
  \end{array}\right],
  \]
    \[
  clip({\bf Q}_1 \star {\bf A}_0, -4, 4) \bar{\bigwedge} {\bf C}_{1} = 
  \left[\begin{array}{rrr}
  \infty & \infty & \infty\\
  -2 & 0 & -4\\
  -4 & -4 & -4\\
  \end{array}\right] 
  \bar{\bigwedge}
    \left[\begin{array}{rrr}
  -4 & -4 & -4\\
  -4 & -4 & -2\\
  -4 & -2 & -4\\
  \end{array}\right] =
  \left[\begin{array}{rrr}
  \infty & \infty & \infty\\
  \infty & \infty & \infty\\
  \infty & \infty & \infty\\
  \end{array}\right],
  \]
  \[{\bf C}_0 = 
    \left[\begin{array}{rrr}
  -4 & -2 & 0\\
  -2 & -4 & 2\\
  0 & 2 & -4\\
  \end{array}\right]
  \bigvee
    \left[\begin{array}{rrr}
  \infty & \infty & \infty\\
  \infty & \infty & \infty\\
  \infty & \infty & \infty\\
  \end{array}\right] =
  \left[\begin{array}{rrr}
  -4 & -2 & 0\\
  -2 & -4 & 2\\
  0 & 2 & -4\\
  \end{array}\right], 
  \]
    \[{\bf P}_0 = {\bf P}_1 \bigvee {\bf Q}_1 = 
  \left[\begin{array}{ccc}
  0 & 2 & 4\\
  2 & 0 & \infty\\
  4 & \infty & 0\\
  \end{array}\right]
  \bigvee
    \left[\begin{array}{rrr}
  \infty & \infty & \infty\\
  \infty & \infty & -2\\
  \infty & -2 & \infty\\
  \end{array}\right] =
  \left[\begin{array}{rrr}
  0 & 2 & 4\\
  2 & 0 & -2\\
  4 & -2 & 0\\
  \end{array}\right], 
  \]
  \[
  {\bf Q}_0 = chop({\bf C}_0, -3, 4) = 
  \left[\begin{array}{rrr}
  \infty & -2 & \infty\\
  -2 & \infty & 2\\
  \infty & 2 & \infty\\
  \end{array}\right].\]
  }
  In the {\bf for} loop in lines $18$ to $20$, we set ${\bf B}_k = ({\bf C}_k \geq 0)$ for $k=1$ and $k=2$.  This means 
\begin{align*}
{\bf B}_1 & = ({\bf C}_1 \geq 0) = \left(
   \left[\begin{array}{rrr}
  -4 & -4 & -4\\
  -4 & -4 & -2\\
  -4 & -2 & -4\\
  \end{array}\right] \geq 0\right) = 
  \left[\begin{array}{rrr}
  0 & 0 & 0\\
  0 & 0 & 0\\
  0 & 0 & 0\\
  \end{array}\right],
\\
{\bf B}_2 & = ({\bf C}_2 \geq 0) = \left(
   \left[\begin{array}{rrr}
  -4 & -4 & -4\\
  -4 & -4 & -4\\
  -4 & -4 & -4\\
  \end{array}\right] \geq 0\right) = 
  \left[\begin{array}{rrr}
  0 & 0 & 0\\
  0 & 0 & 0\\
  0 & 0 & 0\\
  \end{array}\right].
\end{align*} 

  We then set ${\bf B}_0 = (0 \leq {\bf P}_0 < 4)$, ${\bf R} = {\bf P}_0 \mod 4$, 
  and $\mathrm{\Delta} = 4\cdot \sum_{k=0}^{2} 2^{k}\cdot {\bf B}_k + {\bf R}$ according to lines 
  $21$ to $23$:
  {\scriptsize
  \[{\bf B}_0 = (0 \leq {\bf P}_0 < 4) = \left(0 \leq 
   \left[\begin{array}{rrr}
     0 & 2 & 4\\
  2 & 0 & -2\\
  4 & -2 & 0\\
   \end{array}\right] < 4\right) = 
  \left[\begin{array}{rrr}
     1 & 1 & 0\\
  1 & 1 & 0\\
  0 & 0 & 1\\
   \end{array}\right],
  \]
  \[{\bf R} = {\bf P}_0 \mod 4 = 
   \left[\begin{array}{rrr}
     0 & 2 & 4\\
  2 & 0 & -2\\
  4 & -2 & 0\\
   \end{array}\right] \mod 4 = 
    \left[\begin{array}{rrr}
     0 & 2 & 0\\
  2 & 0 & -2\\
  0 & -2 & 0\\
   \end{array}\right],
  \]
  \[\begin{array}{rcl}
  \mathrm{\Delta} & = & 4 \cdot \displaystyle\sum_{k=0}^{2} 2^{k}\cdot {\bf B}_k + {\bf R}\\
   & = & 4\cdot \left(
   2^{0} \cdot \left[\begin{array}{rrr}
    1 & 1 & 0\\
    1 & 1 & 0\\
    0 & 0 & 1\\
   \end{array}\right] + 
   2^{1} \cdot \left[\begin{array}{rrr}
    0 & 0 & 0\\
    0 & 0 & 0\\
    0 & 0 & 0\\
   \end{array}\right] + 
   2^{2} \cdot \left[\begin{array}{rrr}
    0 & 0 & 0\\
    0 & 0 & 0\\
    0 & 0 & 0\\
   \end{array}\right] 
   \right) + 
   \left[\begin{array}{rrr}
    0 & 2 & 0\\
    2 & 0 & -2\\
    0 & -2 & 0\\
   \end{array}\right]\\
   & = & 4\cdot \left(
   \left[\begin{array}{rrr}
    1 & 1 & 0\\
    1 & 1 & 0\\
    0 & 0 & 1\\
   \end{array}\right] + 
   \left[\begin{array}{rrr}
    0 & 0 & 0\\
    0 & 0 & 0\\
    0 & 0 & 0\\
   \end{array}\right] + 
   \left[\begin{array}{rrr}
    0 & 0 & 0\\
    0 & 0 & 0\\
    0 & 0 & 0\\
   \end{array}\right] 
   \right) + 
   \left[\begin{array}{rrr}
    0 & 2 & 0\\
    2 & 0 & -2\\
    0 & -2 & 0\\
   \end{array}\right]\\
   & = & 
   \left[\begin{array}{rrr}
    4 & 4 & 0\\
    4 & 4 & 0\\
    0 & 0 & 4\\
   \end{array}\right] +  
   \left[\begin{array}{rrr}
    0 & 2 & 0\\
    2 & 0 & -2\\
    0 & -2 & 0\\
   \end{array}\right] = 
      \left[\begin{array}{rrr}
    4 & 6 & 0\\
    6 & 4 & -2\\
    0 & -2 & 4\\
   \end{array}\right]\\
   \end{array}
  \]
  }
  The algorithm terminates by returning $\mathrm{\Delta}$ in line $24$.
  The resulting shortest path costs are presented in 
  Table~\ref{table:SW_results}.  
  However, if we examine graph $G'$ (Figure~\ref{fig:example1}), we 
find that these results are \textit{incorrect}.
 The correct shortest paths are provided in Table~\ref{table:correct}.  
Therefore, the Shoshan-Zwick algorithm is incorrect.

\begin{table}[h]
\caption{Shortest paths for graph $G'$ based on the Shoshan-Zwick algorithm.}
\label{table:SW_results}
\centering
\begin{tabular}{|c|rrr|}
\hline
$\mathrm{\Delta}$ & $a$ & $b$ & $c$  \\
\hline
$a$ & $4$ & $6$ & $0$ \\
$b$ & $6$ & $4$ & $-2$ \\
$c$ & $0$ & $-2$ & $4$ \\
\hline
\end{tabular}
\end{table}

\begin{table}[h]
\caption{Correct shortest paths for graph $G'$ .}
\label{table:correct}
\centering
\begin{tabular}{|c|rrr|}
\hline
$\mathrm{\Delta}$ & $a$ & $b$ & $c$  \\
\hline
$a$ & $0$ & $2$ & $4$ \\
$b$ & $2$ & $0$ & $6$ \\
$c$ & $4$ & $6$ & $0$ \\
\hline
\end{tabular}
\end{table}

\section{The Errors in the Algorithm}
\label{error}

In this section, we describe what causes the erroneous behavior of the Shoshan-Zwick algorithm. 
Recall that $\mathrm{\Delta}$ is the matrix that contains the costs of the shortest paths between all pairs of vertices 
after the algorithm terminates.
Moreover, let $\delta(i,j)$ denote the cost of the shortest path between nodes $i$ and $j$.
After the termination of the algorithm, we must have $\mathrm{\Delta}_{ij}=\delta(i,j)$ for any $i,j\in\{1,\ldots ,n\}$. 
However, as we showed in the counter-example in Section \ref{counter},
it may be the case that $\mathrm{\Delta}_{ij}\neq\delta(i,j)$ for some $i,j$ at termination. 
The exact errors of the algorithm are as follows:

\begin{enumerate}[1.]
 \item ${\bf R}$ is not computed correctly.
 \item ${\bf B}_0$ is not computed correctly.
 \item $\mathrm{\Delta}$ is not computed correctly, since  $M\cdot {\bf B}_0 + {\bf R}$ is part of the sum producing it.
\end{enumerate}

\noindent In the rest of this section, we illustrate what causes these errors.

It is clear from line $23$ of Algorithm \ref{alg:sz} that
the matrices ${\bf B}_k$ (for $0\leq k\leq l$) represent the $\lceil \log_2 n\rceil$ most significant 
bits of each distance.  That is,
\[({\bf B}_k)_{ij} = \left\{\begin{array}{ll}
1 & \mbox{if } 2^{k}\cdot M \mbox{ must be added to } \mathrm{\Delta}_{ij} \mbox{ so that } \mathrm{\Delta}_{ij}=\delta(i,j)\\
0 & \mbox{otherwise}
\end{array}\right.,\]
while $\bf R$ represents the remainder of each distance modulo $M$. 
This is also illustrated in \cite[Lemma 3.6]{Shosh1}, where 
for every $0\leq k\leq l$, 
$({\bf B}_k)_{ij} = 1$ if and only if $\delta(i,j) \mod 2^{k+m+1}\geq 2^{k+m}$,
while ${\bf R}_{ij} = \delta(i,j) \mod M$.
Hence, for every $i,j$, we must have 
\begin{eqnarray}\label{InitialSum}
(M\cdot {\bf B}_0 + {\bf R})_{ij} = \delta(i,j) \mod 2^{m+1}.
\end{eqnarray}

The first error of the algorithm arises immediately from the key observation that ${\bf P}_0$ can 
have entries with negative values. 
This means that line $22$ (that sets ${\bf R}_{ij} = ({\bf P}_0)_{ij}  \mod M$) 
is not correctly calculating ${\bf R}_{ij} = \delta(i,j) \mod M$ 
since $\delta(i,j)\geq 0$ by definition, while $({\bf P}_0)_{ij}$ can be negative.

A closer examination of how ${\bf P}_0$ obtains its negative values reveals another error of the algorithm.
The following definitions are given in \cite[Section 3]{Shosh1}. 
Consider a set $Y\subseteq [0,M\cdot n]$. Note that $[0,M\cdot n]$ includes any value that $\delta(i,j)$ can take, since 
$n$ is the number of nodes, and $M$ is the maximum edge cost. 
Let $Y=\cup_{r=1}^p[a_r,b_r]$, where $a_r\leq b_r$, for $1\leq r\leq p$ and $b_r < a_{r+1}$, for $1\leq r < p$.
Let ${\bf \Delta}_Y$ be an $n\times n$ matrix, whose elements are in the range $\{-M,\ldots ,M\}\cup\{+\infty\}$, such that 
for every $1\leq i,j \leq n$, we have
{\small
\begin{eqnarray}\label{Y-ij}
({\bf \Delta}_Y)_{ij} = \left\{\begin{array}{ll}
 -M & \mbox{ if } a_r \leq \delta(i,j) \leq b_r-M \mbox{ for some } 1\leq r\leq p, \\     
\delta(i,j) - b_r & \mbox{ if } b_r-M < \delta(i,j) \leq b_r+M \mbox{ for some } 1\leq r\leq p,\\
 +\infty & \mbox{ otherwise.} 
\end{array}\right.
\end{eqnarray}
}

By \cite[Lemma 3.5]{Shosh1}, ${\bf P}_0 = {\bf \Delta}_{Y_{0}}$, where 
$Y_0 = \{x | (x \mod 2^{m+1}) = 0\}$. Recall that $2^m = M$.
Note that by definition of $Y_0$, when calculating ${\bf P}_0 = {\bf \Delta}_{Y_{0}}$,
it can only be the case that $a_r=b_r$.
Moreover, $b_r = 2^{m+1}\cdot (r-1)$ for $1 \leq r \leq p$, where $p$ is such that
$2^{m+1}\cdot (p-1)\leq M\cdot n < 2^{m+1}\cdot p$.
But then: 
\[\left(\cup_{r=1}^p[b_r-M,b_r+M]\right)\supset [0,M\cdot n]\]
That is, $(\cup_{r=1}^p[b_r-M,b_r+M])$ covers all possible values that $\delta(i,j)$ may take 
for any $i,j$. Hence,
{\footnotesize
\begin{eqnarray}\label{P_0-A}
({\bf P}_0)_{ij}= \left\{\begin{array}{ll}
\delta(i,j) & \mbox{for } r=1 \mbox{ (i.e., if } \delta(i,j) \leq 2^m),\\     
\delta(i,j) - b_r & \mbox{for } 2 \leq r \leq p, \mbox{ such that } b_r-2^m < \delta(i,j) \leq b_r+2^m.\\
\end{array}\right.
\end{eqnarray}
}
Let us examine the values that $({\bf P}_0)_{ij}$ takes by equation (\ref{P_0-A}):
\begin{itemize}
\item For $0\leq\delta(i,j)\leq 2^{m}$, we have $({\bf P}_0)_{ij}=\delta(i,j)\mod 2^{m+1}$.
\item For $2^m < \delta(i,j) < 2^{m}+2^m$, we have $({\bf P}_0)_{ij}=(\delta(i,j)\mod 2^{m+1}) - 2^{m+1}$.
\item For $2^{m+1} \leq\delta(i,j) \leq 2^{m+1}+2^m$, we have $({\bf P}_0)_{ij}=\delta(i,j)\mod 2^{m+1}$.
\item For $2^{m+1}+2^m < \delta(i,j) < 2^{m+2}+2^m$, we have $({\bf P}_0)_{ij}=(\delta(i,j)\mod 2^{m+1}) - 2^{m+1}$.
\item And so forth...
\end{itemize} 
More formally, equation (\ref{P_0-A}) can be rewritten as follows:
\begin{eqnarray}\label{P_0-B}
({\bf P}_0)_{ij}= \left\{\begin{array}{ll}
\delta(i,j) \mod 2^{m+1} & \mbox{if } \delta(i,j) \mod 2^{m+1} \leq 2^{m},\\
(\delta(i,j) \mod 2^{m+1}) - 2^{m+1} & \mbox{if } \delta(i,j) \mod 2^{m+1} > 2^{m}.\\     
\end{array}\right.
\end{eqnarray}
Moreover, equation (\ref{P_0-B}) implies that 
\begin{eqnarray}\label{P_0-2}
\text{for } i,j \text{ such that } \delta(i,j) \mod 2^{m+1} \leq 2^{m} \text{, we have } 0 \leq ({\bf P}_{0})_{ij} \leq M,
\end{eqnarray}
while
\begin{eqnarray}\label{P_0-1}
\text{for } i,j \text{ such that } \delta(i,j) \mod 2^{m+1} > 2^{m} \text{, we have } -M < ({\bf P}_0)_{ij} < 0.
\end{eqnarray}
Recall now that we must have $({\bf B}_0)_{ij} = 1$ if and only if $\delta(i,j) \mod 2^{m+1}\geq 2^{m}$.
However, from equations (\ref{P_0-2}) and (\ref{P_0-1}), 
this does not hold (as claimed in the proof of \cite[Lemma 3.6]{Shosh1}) for ${\bf B}_0=(0 \leq {\bf P}_0 < M)$ 
(i.e., line $21$ of Algorithm \ref{alg:sz}). Therefore, the algorithm does not compute ${\bf B}_0$ correctly.

It is clear that in the presence of these two identified errors (in calculating $\bf R$ and ${\bf B}_0$), 
the algorithm is not computing $\mathrm{\Delta}$ correctly. 

To accommodate understanding, 
let us consider the case of ${\bf P}_0$ for graph $G'$ (Figure \ref{fig:example1}) in Section \ref{counter}.
We have $Y_0\subseteq [0, M\cdot n]$, i.e., $Y_0\subseteq [0, 12]$.
Further, $Y_0 = \{x | (x \mod 2^{2+1}) = 0\}$. This means that $Y_0 = \{0,8\}$.
Thus, with respect to the definition of $({\bf \Delta}_Y)_{ij}$,
we have $a_1=b_1=0$ and $a_2=b_2=8$. Therefore,
{\small
\[({\bf P}_0)_{ij}=({\bf \Delta}_{Y_0})_{ij} = \left\{\begin{array}{ll}
 \delta(i,j) & \mbox{ if } -4 < \delta(i,j) \leq 4 \mbox{ (for } r=1 \mbox{ and hence } b_r=0),\\     
\delta(i,j) - 8 & \mbox{ if } 4 < \delta(i,j) \leq 12 \mbox{ (for } r=2 \mbox{ and hence } b_r=8).\\
\end{array}\right.
\]
}
Let us consider the shortest path cost of nodes $a$ ($i=1$) and $b$ ($j=2$). We have $\delta(1,2)=2$, and thus
$({\bf P}_0)_{12}=2$. 
We now consider the shortest path cost of nodes $b$ ($i=2$) and $c$ ($j=3$). We then have $\delta(2,3)=6$. 
This, however, means that $({\bf P}_0)_{23}=-2$. 
Although $\delta(2,3)\mod 2^3 > 2^2$, we have $({\bf B}_0)_{23} = 0$. 
Moreover, $\bf R_{23}=({\bf P}_0)_{23} \mod 8 = -2$ (instead of $\bf R_{23}=\delta(2,3) \mod 8 = 6$). 
These two errors lead to $\mathrm{\Delta}_{23}= -2$ instead of $6$ (see Table \ref{table:correct}).

  \section{The Revised Algorithm}
  \label{corr}
  
  In this section, we present a new version of the Shoshan-Zwick algorithm 
  that resolves the problems illustrated in Section~\ref{error}.  
  Lines $1$ to $20$ of Algorithm~\ref{alg:sz} remain unchanged.  
  We make the following changes to lines $21$ to $24$:
  
  \begin{enumerate}
   \item We replace ${\bf B}_{0}$ with ${\bf \hat{B}}_{0}$ and set ${\bf \hat{B}}_{0}$ to $(-M < {\bf P}_0 < 0)$ in line $21$.
   \item We replace $\bf R$ with ${\bf \hat{R}}$ and set ${\bf \hat{R}}$ to ${\bf P}_0$ in line $22$.
   \item We set $\mathrm{\Delta}$ to $M\cdot \displaystyle\sum_{k=1}^{l} 2^k \cdot {\bf B}_k + 2\cdot M\cdot {\bf \hat{B}}_0 + {\bf \hat{R}}$ in line $23$.
  \end{enumerate}
  
  Note that we have replaced ${\bf B}_{0}$ and $\bf R$ with ${\bf \hat{B}}_{0}$ and ${\bf \hat{R}}$, respectively.  
The purpose of the change in notation is to show that these matrices  no longer represent the incorrect 
versions from the original (erroneous) algorithm.
 Lines $21$ to $24$ of the revised algorithm are illustrated in Algorithm \ref{alg:sz_correct}.

   \begin{algorithm}[htbp]
   21: {\hspace{.025em}} ${\bf \hat{B}}_{0} = (-M < {\bf P}_0 < 0)$.\\
   22: {\hspace{.025em}} ${\bf \hat{R}} = {\bf P}_0$.\\
   23: {\hspace{.025em}} $\mathrm{\Delta} = M\cdot \displaystyle\sum_{k=1}^{l} 2^k \cdot {\bf B}_k + 2\cdot M\cdot {\bf \hat{B}}_0 + {\bf \hat{R}}$. \\
   24: {\hspace{.025em}} {\bf return} $\mathrm{\Delta}$.
  \caption{Corrected portion of lines $21$ to $24$.}
  \label{alg:sz_correct}
\end{algorithm}

We refer to our counter-example in Section~\ref{counter}.  
Since the algorithm is correct up to line $20$, we will examine how the algorithm operates in lines $21$ to $24$.  
Hence, we set ${\bf \hat{B}}_0 = (-4 < {\bf P}_0 < 0)$, ${\bf \hat{R}} = {\bf P}_0$, 
and $\mathrm{\Delta} = 4\cdot \sum_{k=1}^{2} 2^{k}\cdot {\bf B}_k + 8\cdot {\bf \hat{B}}_0 + {\bf \hat{R}}$.  
{\scriptsize
  \[{\bf \hat{B}}_0 = (-4 < {\bf P}_0 < 0) = \left(-4 \leq 
   \left[\begin{array}{rrr}
     0 & 2 & 4\\
  2 & 0 & -2\\
  4 & -2 & 0\\
   \end{array}\right] < 0\right) = 
  \left[\begin{array}{rrr}
     0 & 0 & 0\\
  0 & 0 & 1\\
  0 & 1 & 0\\
   \end{array}\right],
  \]
  
  \[{\bf \hat{R}} = {\bf P}_0 = 
   \left[\begin{array}{rrr}
     0 & 2 & 4\\
  2 & 0 & -2\\
  4 & -2 & 0\\
   \end{array}\right],
  \]
  
  \[\begin{array}{rcl}
  \mathrm{\Delta} & = & 4 \cdot \displaystyle\sum_{k=1}^{2} 2^{k}\cdot {\bf B}_k + 8\cdot {\bf \hat{B}}_0 + {\bf \hat{R}}\\
  
   & = & 4\cdot \left(
   2^{1} \cdot \left[\begin{array}{rrr}
  0 & 0 & 0\\
  0 & 0 & 0\\
  0 & 0 & 0\\
   \end{array}\right] + 
   2^{2}\cdot \left[\begin{array}{rrr}
  0 & 0 & 0\\
  0 & 0 & 0\\
  0 & 0 & 0\\
  \end{array}\right]
   \right) + 
   8 \cdot 
   \left[\begin{array}{rrr}
     0 & 0 & 0\\
  0 & 0 & 1\\
  0 & 1 & 0\\
   \end{array}\right] + 
      \left[\begin{array}{rrr}
     0 & 2 & 4\\
  2 & 0 & -2\\
  4 & -2 & 0\\
   \end{array}\right] \\     
  & = &\left[\begin{array}{rrr}
    0 & 2 & 4\\
    2 & 0 & 6\\
    4 & 6 & 0\\
   \end{array}\right]\\
 
   \end{array}
  \]
  }
  Line $24$ returns $\mathrm{\Delta}$. The elements of $\mathrm{\Delta}$ reflect the correct shortest path costs given in Table~\ref{table:correct}. Therefore, 
  the revised algorithm works correctly for our counter-example.  The corrections in the algorithm are not limited to the 
  counter-example, as shown in the theorem below:
  
  \begin{theorem}
The revised Shoshan-Zwick algorithm calculates all the shortest path costs in an undirected graph 
with integer edge costs in the range $\{1,\ldots , M\}$. 
\end{theorem}
\begin{proof}
It suffices to show that $2\cdot M\cdot {\bf \hat{B}}_0 + {\bf \hat{R}}$
represents what the original algorithm intended to represent with $M\cdot {\bf B}_0 + {\bf R}$.
That is, by equation (\ref{InitialSum}), it suffices to show that 
$(2\cdot M\cdot {\bf \hat{B}}_0 + {\bf \hat{R}})_{ij} = \delta(i,j) \mod 2^{m+1}$, for every 
$1\leq i,j\leq n$.

First, we consider the case where $\delta(i,j) \mod 2^{m+1} \leq 2^{m}$.
Equation (\ref{P_0-2}) indicates that $0 \leq ({\bf P}_{0})_{ij} \leq M$.
Hence, $({\bf \hat{B}}_0)_{ij}=0$ (by line $21$ of the revised algorithm).
Moreover, since ${\bf \hat{R}}_{ij}= ({\bf P}_{0})_{ij}$ (by line $22$ of the revised algorithm),
we have that ${\bf \hat{R}}_{ij}= \delta(i,j)\mod 2^{m+1}$ by equation (\ref{P_0-B}).
Thus, $(2\cdot M\cdot {\bf \hat{B}}_0 + {\bf \hat{R}})_{ij} = \delta(i,j) \mod 2^{m+1}$.

We next consider the case where $\delta(i,j) \mod 2^{m+1} > 2^{m}$. 
Equation (\ref{P_0-1}) indicates that $-M < ({\bf P}_0)_{ij} < 0$.
Hence, $({\bf \hat{B}}_0)_{ij}=1$ (by line $21$ of the revised algorithm).
Further, ${\bf \hat{R}}_{ij}= (\delta(i,j) \mod 2^{m+1}) - 2^{m+1}$ by equation (\ref{P_0-B}).
Therefore, $(2\cdot M\cdot {\bf \hat{B}}_0 + {\bf \hat{R}})_{ij} = 2\cdot 2^{m}\cdot 1 + (\delta(i,j) \mod 2^{m+1}) - 2^{m+1} = 
\delta(i,j) \mod 2^{m+1}$,
which completes the proof.
\end{proof}
 
  \section{Efficacy}
  \label{eff}
  
  In this section, we identify issues with implementing the Shoshan-Zwick algorithm.  Recall that the algorithm 
  runs in $\tilde{O}(M\cdot n^{\omega})$ time, where $\omega$ is the exponent for the fastest known matrix 
  multiplication algorithm.    This means that the running time of the algorithm depends on which matrix multiplication 
  algorithm is used. We will discuss two subcubic matrix multiplication algorithms and show why it is impractical to use them in the Shoshan-Zwick implementation.
  
  The current fastest matrix multiplication algorithm is provided by \cite{Williams12}, where 
  $\omega < 2.3727$.  This approach tightens the techniques used in \cite{CW87}, which gives a matrix multiplication 
  algorithm where $\omega < 2.376$.  Although both matrix multiplication algorithms are theoretically 
  efficient, neither one is practical to implement.  They both provide an advantage only for matrices that are too 
  large for modern hardware to handle \cite{Robinson05}.

 We next consider Strassen's matrix multiplication algorithm \cite{Strassen69}, which runs in $O(n^{2.8074})$ 
time.  Recall that the algorithm computes ${\bf A} \cdot {\bf B} = {\bf C}$ by partitioning the matrices into equally 
sized block matrices
\[{\bf A} = \left[\begin{array}{rr}
{\bf A}_{1,1} & {\bf A}_{1,2}\\
{\bf A}_{2,1} & {\bf A}_{2,2}\\
\end{array}\right],
{\bf B} = \left[\begin{array}{rr}
{\bf B}_{1,1} & {\bf B}_{1,2}\\
{\bf B}_{2,1} & {\bf B}_{2,2}\\
\end{array}\right],
{\bf C} = \left[\begin{array}{rr}
{\bf C}_{1,1} & {\bf C}_{1,2}\\
{\bf C}_{2,1} & {\bf C}_{2,2}\\
\end{array}\right]\]

where 

\[
\begin{array}{rcl}
{\bf C}_{1,1} & = & {\bf A}_{1,1}\cdot {\bf B}_{1,1} + {\bf A}_{1,2}\cdot {\bf B}_{2,1}\\
{\bf C}_{1,2} & = & {\bf A}_{1,1}\cdot {\bf B}_{1,2} + {\bf A}_{1,2}\cdot {\bf B}_{2,2}\\
{\bf C}_{2,1} & = & {\bf A}_{2,1}\cdot {\bf B}_{1,1} + {\bf A}_{2,2}\cdot {\bf B}_{2,1}\\
{\bf C}_{2,2} & = & {\bf A}_{2,1}\cdot {\bf B}_{1,2} + {\bf A}_{2,2}\cdot {\bf B}_{2,2}\\
\end{array}
\]

To reduce the total number of multiplications, seven new matrices are defined as follows:

\[
\begin{array}{rcl}
{\bf M}_{1} & = & ({\bf A}_{1,1} + {\bf A}_{2,2})\cdot ({\bf B}_{1,1} + {\bf B}_{2,2})\\
{\bf M}_{2} & = & ({\bf A}_{2,1} + {\bf A}_{2,2})\cdot {\bf B}_{1,1}\\
{\bf M}_{3} & = & {\bf A}_{1,1} \cdot ({\bf B}_{1,2} - {\bf B}_{2,2})\\
{\bf M}_{4} & = & {\bf A}_{2,2} \cdot ({\bf B}_{2,1} - {\bf B}_{1,1})\\
{\bf M}_{5} & = & ({\bf A}_{1,1} + {\bf A}_{1,2})\cdot {\bf B}_{2,2}\\
{\bf M}_{6} & = & ({\bf A}_{2,1} - {\bf A}_{1,1})\cdot ({\bf B}_{1,1} + {\bf B}_{1,2})\\
{\bf M}_{7} & = & ({\bf A}_{1,2} - {\bf A}_{2,2})\cdot ({\bf B}_{2,1} + {\bf B}_{2,2})\\
\end{array}
\]

With these new matrices, the block matrices of ${\bf C}$ can be redefined as 
\[
\begin{array}{rcl}
{\bf C}_{1,1} & = & {\bf M}_1 + {\bf M}_4 - {\bf M}_5 + {\bf M}_7\\
{\bf C}_{1,2} & = & {\bf M}_3 + {\bf M}_5\\
{\bf C}_{2,1} & = & {\bf M}_2 + {\bf M}_4\\
{\bf C}_{2,2} & = & {\bf M}_1 - {\bf M}_2 + {\bf M}_3 + {\bf M}_6\\
\end{array}
\]

The process of dividing ${\bf C}$ repeats recursively $n$ times until the submatrices degenerate into a single number.

Although Strassen's algorithm is faster than the naive $O(n^3)$ matrix multiplication algorithm, we cannot directly use it 
in the Shoshan-Zwick algorithm.  Recall that the naive approach uses matrix multiplication over the closed semi-ring $\{+, \cdot\}$.  The 
Shoshan-Zwick algorithm, on the other hand, actually uses matrix multiplication 
over the closed semi-ring $\{\min,+\}$, which is known as ``funny matrix multiplication'' or the ``distance product'' 
in the literature.  Note that the sum operation in the naive approach is equivalent to the $\min$ operation in ``funny matrix 
multiplication''.  However, Strassen's algorithm requires an additive inverse. This implies that an inverse for the $\min$ operation is 
needed in ``funny matrix multiplication''. Such an inverse does not exist.  In fact, we cannot multiply matrices with 
less than $\Omega(n^3)$ operations when only the $\min$ and sum operations are allowed \cite{AGM97,kerr70}.  
Thus, Strassen's algorithm cannot directly be used for 
computing shortest paths.

One potential solution is to encode a matrix used for distance products such that regular matrix multiplication works.  
\cite{AGM97} provides an approach for this conversion as follows:  Suppose we want to convert an $n \times n$ matrix 
${\bf A}$ to ${\bf A}'$.  We set \[a'_{ij} = (n+1)^{a_{ij}-x},\] where $x$ is the smallest value in ${\bf A}$.  
We perform the same conversion from ${\bf B}$ to ${\bf B}'$.  We then obtain ${\bf C}' = {\bf A}'\cdot{\bf B}'$ by:
\[c'_{ij} = \displaystyle\sum_{k=1}^{n} (n+1)^{a_{ik}+b_{kj}-2\cdot x}.\]  We then use binary search to find the 
largest $s_{ij}$ such that $s_{ij} \leq a'_{ik}+b'_{kj}-2\cdot x$, and we set $c_{ij} = s_{ij} + 2\cdot x$.  This 
gives us ${\bf C}$, which is the distance product of matrices ${\bf A}$ and ${\bf B}$.

\cite{AGM97} states that the algorithm performs $O(n^{\omega})$ operations on integers which are $\leq n\cdot (n+1)^{2\cdot M}$, 
where $M$ is the magnitude of the largest number.  For large numbers, we would need $O(M\cdot \log M)$ operations on 
$O(\log n)$-bit numbers.  As a result, the total time to compute the distance product is $O(M\cdot n^{\omega}\cdot \log M)$.  If we 
apply this to the Shoshan-Zwick algorithm, the algorithm takes $\tilde{O}(M^{2}\cdot n^{\omega}\cdot \log M)$ time.

Although the above algorithm provides a subcubic approach for implementing the Shoshan-Zwick algorithm, it is not necessarily 
the most efficient implementation.  This is because there exist other efficient APSP algorithms for integer edge costs.  
For instance, we can implement 
Goldberg's $O(m+n\cdot \log N)$ time single source shortest path algorithm \cite{Gol01b}, where $N$ is the largest 
edge cost, and run it $n$ times; one for each vertex.  Goldberg's implementation is one of the currently known 
fastest implementations available.  The resulting implementation is substantially faster compared 
to the Shoshan-Zwick algorithm.

\section{Conclusion}
\label{conc}

In this paper, we revised the Shoshan-Zwick algorithm to resolve
issues related to its correctness.  We first 
provided a counter-example which shows that the algorithm is incorrect.  We then identified the exact cause of the 
problem and presented a modified algorithm that resolves the problem.  We also explained the efficacy issues 
of the algorithm. 
An interesting study would be to implement the Shoshan-Zwick algorithm and profile it with efficient APSP algorithms for 
graphs with integer edge costs.  

\bibliographystyle{model1-num-names}

\end{document}